\documentclass[%
 reprint,
 amsmath,amssymb,
 aip,
]{revtex4-2}

\usepackage{enumitem}
\usepackage{graphicx}
\usepackage{dcolumn}
\usepackage{bm}
\usepackage{comment}
\usepackage{graphicx}
\usepackage{dcolumn}
\usepackage{bm}
\usepackage{hyperref}
\usepackage{caption}
\usepackage{amssymb}
\usepackage{amsmath}
\usepackage[dvipsnames]{xcolor}
\usepackage{soul}

\usepackage{dcolumn}
\usepackage{bm}
\usepackage{float}
\usepackage{bm}
\def \bF {\pmb{F}}

\def \bD {\pmb{D}}

\def \bu {\pmb{u}}

\def \bx {\pmb{x}}

\def \bX {\pmb{X}}

\def \bs {\pmb{s}}

\usepackage{rotating}
\usepackage[english]{babel}
\usepackage{amsthm}

\newtheorem{remark}{Remark}

\newtheorem{proposition}{Proposition}

\usepackage{subfig}

\makeatletter
\renewcommand{\fnum@figure}{FIG. \thefigure}
\makeatother

\makeatletter
\providecommand{\@LN}[2]{}
\makeatother

\DeclareMathOperator{\sign}{sign}

\begin{document}




\title{The efficiency of synchronization dynamics and the role of network syncreactivity}

\author{Amirhossein Nazerian}
\affiliation{Mechanical Engineering Department, University of New Mexico, Albuquerque, NM 87131, USA}
\author{Joseph D Hart}
\affiliation{US Naval Research Laboratory, Code 5675, Washington, DC 20375, USA}
\author{Matteo Lodi}
\affiliation{DITEN, University of Genoa, Via Opera Pia 11a, I-16145, Genova, Italy}
\author{Francesco Sorrentino$^1$ \\ Email: \href{mailto:fsorrent@unm.edu}{fsorrent@unm.edu}}
\begin{abstract}
{Synchronization of coupled oscillators is a fundamental process in both natural and artificial networks. While much work has investigated the asymptotic stability of the synchronous solution, the fundamental question of the transient behavior toward synchronization has received far less attention. In this work, we present the transverse reactivity as a metric to quantify the instantaneous rate of growth or decay of desynchronizing perturbations. 
We first use the transverse reactivity to design a coupling-efficient and energy-efficient synchronization strategy that involves varying the coupling strength dynamically according to the current state of the system. We find that our synchronization strategy is able to synchronize networks in both simulation and experiment over a significantly larger (often by orders of magnitude) range of coupling strengths than is possible when the coupling strength is constant. Then, we characterize the effects of network topology on the transient dynamics towards synchronization by introducing the concept of network syncreactivity: A network with a larger syncreactivity has a larger transverse reactivity at every point on the synchronization manifold, independent of the oscillator dynamics. We classify real-world examples of complex networks in terms of their syncreactivity.}

\end{abstract}
\maketitle

\section{Introduction}

The synchronization of networks of coupled oscillators has been the subject of intensive investigation, see e.g.,  
\cite{arenas2008synchronization,Report,wu2007synchronization,suykens2008introduction}.  Compared to the analysis of stability of the synchronous solution, the question of the efficiency of the synchronization dynamics has received less attention \cite{locatelli2015efficient,moujahid2011efficient,zhang2020energy,urazhdin2010fractional,murali1993transmission,demidov2014synchronization}.
 However, all biological and technological systems must operate efficiently.
In addition, these systems do not typically communicate at all times but often
interact in a state-dependent fashion.
For example, neurons in the brain transmit signals to other neurons after they `fire' \cite{izhikevich2007dynamical}, and  similar activation mechanisms have been found to describe
interactions among fireflies synchronizing their flashing 
\cite{mccrea2022model,ramirez2019modeling} and the way pacemaker cells in the heart affect surrounding cells 
via short action potentials separated by long depolarization bouts  \cite{verkerk2007single,dokos1996ion}.
{
Our goal in this paper is to design a synchronization strategy that is coupling-efficient and energy-efficient; 
i.e., 
it requires lower coupling strength on average and lower synchronization energy in comparison to the case of constant coupling.
}
We show how coupling-efficiency and energy-efficiency of the synchronization dynamics can be achieved by a 
strategy 
which uses coupling only when needed, where the coupling strength 
is varied based on the specific regions of the attractor on which the synchronous solution evolves. 
In particular, we identify a property, the transverse reactivity of different points on the attractor, based on which we adjust the coupling strength.


Fundamental works have characterized the asymptotic stability of the synchronous solution {by exploiting different tools. Algebraic conditions can be found for simple networks of phase oscillators \cite{kuramoto1975international,kuramoto1984chemical,sepulchre2005collective,ha2010complete}; these methods can be extended to more complex oscillators by using the phase resetting curve \cite{brown2004phase} and to its extension \cite{wilson2016isostable,wilson2020phase}. Another widely used tool to compute the asymptotic stability of the synchronous solution is the master stability function (MSF) \cite{SYNCBOOK,Report,Pe:Ca,Ba:Pe02}, which employs the Lyapunov exponent to evaluate the asymptotic rate of growth or decay of perturbations transverse to the synchronous solution.  
For a given choice of oscillator and coupling function, the MSF evaluates the asymptotic stability of the synchronous solution} in 4 steps: (i) describe the time evolution of perturbations of the network trajectory from the synchronous state, (ii) linearize the equations that describe the perturbations' time evolutions, (iii) separate the perturbations parallel to the synchronization manifold from the ones transverse to it, and  (iv) evaluate the rate of growth or decay of the transverse perturbations through the maximum transverse Lyapunov exponent.  The MSF is the function that calculates the maximum transverse Lyapunov as a function of the coupling strength and the network connectivity; therefore, it can be used to identify intervals of the coupling strength within which the synchronous solution is asymptotically stable for a given network topology. 

{Despite the abundance of work and tools in the study of the asymptotic stability of the synchronous solution, there is a lack of methods that study the transient dynamic toward synchronization.}

An important characterization of the transient dynamics of a system is given by the reactivity \cite{Farrell1996Generalized,Neubert1997ALTERNATIVES,Hennequin2012Nonnormal,Tang2014Reactivity,Biancalani2017Giant,Asllani2018Structure,MUOLO2019Patterns,Gudowska2020From,Lindmark2021Centrality,Duan2022Network,nazerian2023single,Nazerian2023Reactability}, which measures the instantaneous rate of growth or decay of the norm of the state vector. The reactivity 
can be thought of as an `instantaneous' finite time Lyapunov exponent \cite{pikovsky2016lyapunov}.  
However, the impact of the reactivity on the synchronization dynamics of complex networks of coupled oscillators is poorly understood. Indeed, while References \cite{Asllani2018Structure,Duan2022Network}  studied the effects of the reactivity on the stability of equilibrium points in networks of coupled dynamical systems, there has been no characterization of the reactivity for the general case of synchronization dynamics, in which all oscillators converge to a synchronous trajectory that evolves in time. 
In this work, we develop a general approach {that employs the MSF paradigm} to evaluate the transverse reactivity that can be applied to a broad variety of systems {of identical oscillators}, and we introduce the `syncreactivity' as an index of the reactivity of the synchronous solution that relates solely to the network topology. { A network with a greater syncreactivity has a larger transverse reactivity at every point on the synchronization manifold, independent of the oscillator dynamics. We show an important link between syncreactivity and normality of the network: Normal networks have minimal syncreactivity.}

{ 
Notice that the proposed method employs the MSF paradigm only to study the transient dynamics and therefore it can be used along with every tool able to analyze the asymptotic synchronization of a network (PRC and others).}

A surprising outcome of our work is that by adjusting the coupling strength according to the instantaneous transverse reactivity, we achieve synchronization, both numerically and experimentally, over intervals of the average coupling strength that are significantly broader than those predicted by the MSF analysis for constant \cite{Pe:Ca} or rapidly time-varying \cite{stilwell2006sufficient} coupling.
In particular, we show that we can significantly lower the minimum coupling strength needed for synchronization, and consequently, the energy expenditure required for synchronization.
In both natural and artificial networks, this has important benefits regarding the actuators that can be used to achieve synchronization, as these are typically limited in the duration and the overall intensity of the coupling they can exert.

Overall, we find that combining transient information provided by the transverse reactivity with traditional asymptotic stability analysis provides an exhaustive characterization of the synchronization dynamics in complex networks. As we will show, our work has broad applications which include linear consensus \cite{Consensus2004Olfati,olfati2007consensus,MIAO2016Collision,LIU2009Consensus,nazerian2023single}, nonlinear control of networked systems \cite{liu2011controllability,yan2015spectrum,PC,panahi2022pinning}, and the control of extreme events and dragon kings in noisy systems \cite{Cavalcante2013Predictability}.
 

\section{Results}

\subsection{Transient Synchronization dynamics and transverse reactivity}


\begin{figure*}[t]
    \centering
    \includegraphics[width=\linewidth]{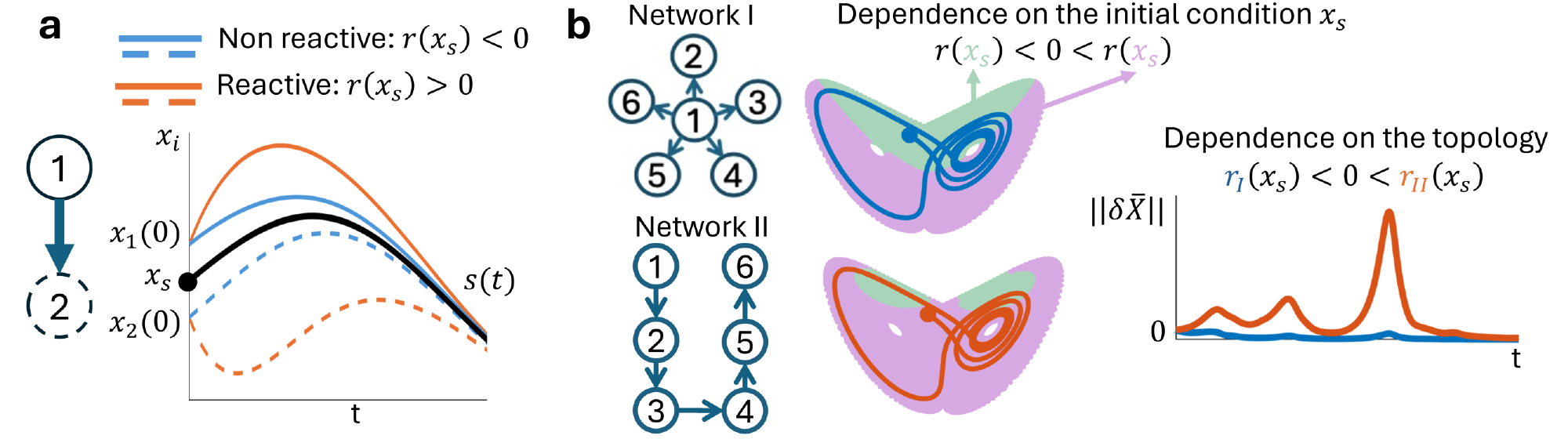}
    \caption{{
    \textbf{The reactivity of the synchronization dynamics of networks.} Panel \textbf{a} illustrates the effects of reactivity on the transient dynamics towards synchronization {of a simple 2-node network. For the set of initial conditions $[x_1(0),x_2(0)]$ close to a point $x_s$ (black dot) belonging to the synchronous solution $s(t)$ (black lines), a non-reactive coupling ($r(x_s) < 0$) results in a direct convergence to the synchronous state (blue lines), while a reactive coupling ($r(x_s) > 0$) results in an initial increase in the separation between trajectories before an eventual settling down to the synchronous state in the long term (orange lines).} Panel \textbf{b} presents two unweighted networks of {y-coupled Lorenz} oscillators {(see Methods \ref{sec:exmethods})} with two different topologies { (Network I and II). Reactive (non-reactive) points $x_s$ on the attractor $\mathcal{A}$ are depicted in lavender (green); therefore, the reactivity of green points is lower than the reactivity computed in lavender points.} {Based on the reactive points density in the attractor}, we can say that Network II is more `reactive' than  Network I.     {Indeed, we see the occurrence of jumps in the distance from the synchronous state $\| \delta \bar{\bX}(t) \|$ (and therefore in the transverse motion to the synchronization manifold)} for the more reactive network II (evolution in red), but not for network I (evolution in blue). }}
    \label{fig:main}
\end{figure*}

We consider {directed} networks of diffusively coupled homogeneous oscillators. 
The network is described by the adjacency matrix $A=[A_{ij}]$, where ${A_{ij}}\geq 0$ measures the strength of the directed coupling from node $j$ to node $i$ ($A_{ij}=0$ if there is no coupling from node $j$ to node $i$.) 
The equations that describe the {dynamics of the network nodes/oscillators are,}
\begin{equation} \label{eq:main}
    \dot{\bx}_i(t)=\bF(\bx_i(t))+\sigma \sum_{j=1}^N L_{ij}H \bx_j(t), \quad i=1,\hdots ,N,
\end{equation}
where {$N$ is the number of nodes/oscillators}, $\bx_i(t)$ is the $n$-dimensional state of node/oscillator $i$ at time $t$. The individual oscillator dynamics is given by $\bF(\bx_i(t))$, and the node-to-node coupling interaction is described by the symmetric and positive semidefinite matrix $H$.
{
The more general case that the node-to-node coupling function is nonlinear is studied in Supplementary Note 1.
}
The Laplacian matrix is denoted by $L=[L_{ij}]$ where {$L_{ij}= A_{ij} - \delta_{ij} \sum_k A_{ik}$}, and $\delta_{ij}$ is the Kronecker delta. By construction, $\sum_j^N L_{ij} = 0, \forall i$. The scalar $\sigma \geq 0$ is the coupling strength. {
We proceed under the assumption (which is required for the stability of the synchronous solution) that the network 
has a directed spanning tree \cite{wu2005algebraic}.
}
Then, the Laplacian matrix has the set of eigenvalues $\{\lambda_i\} $, of which $\lambda_1=0$ and all the others have negative real parts. Moreover, {$Re(\lambda_1) > Re(\lambda_2) \geq... \geq Re(\lambda_N)$}, where {$Re(\cdot)$} notation indicates real part.

Equation \eqref{eq:main} admits the synchronous solution $\bx_1(t)=\bx_2(t)=...=\bx_N(t)=\bs(t)$, {which obeys the dynamics of a single uncoupled system,}
\begin{equation} \label{eq:xss}
\dot{\bs}(t)=\bF(\bs(t)),
\end{equation}
which is independent of the coupling strength $\sigma$, the Laplacian matrix $L$, and the node-to-node coupling matrix $H$. We call $\mathcal{A}$ the attractor on which the dynamics \eqref{eq:xss} converges. {This may be a chaotic attractor.} 

{The transverse reactivity {$r$} measures the instantaneous rate of growth {($r>0$)} or decay {($r<0$)} of the norm of the state vector of desynchronizing perturbations, and can be thought of as an `instantaneous' finite-time Lyapunov exponent \cite{pikovsky2016lyapunov}. This relation to finite-time Lyapunov exponents is explained in detail in Supplementary Note 2.}
{A point $x_s$ on the synchronous solution $s(t)$ is said to be reactive (non-reactive) if $r(x_s) > 0$ ($r(x_s) < 0$).}
Figure \ref{fig:main} illustrates the concept of transverse reactivity. Panel \textbf{a} shows how the transverse reactivity affects the transient dynamics towards synchronization when synchronization is asymptotically stable. For a given system and set of initial conditions, a non-reactive coupling results in a direct convergence to the synchronous state, while a reactive coupling results in an initial increase in the separation between trajectories before an eventual settling down to the synchronous state in the long term.
Panel \textbf{b} presents two unweighted network topologies, one of which is more ‘reactive’ than the other. We consider a network of $y$-coupled Lorenz oscillators {(see Methods \ref{sec:exmethods} and Supplementary Note 3)} and color the attractor lavender (green) to indicate points for which the synchronous dynamics are reactive (non-reactive). It is noteworthy that the reactive part of the attractor is larger for the more reactive network than it is for the less reactive one. We compare the time evolutions of {the system defined on} the two networks starting from the same initial condition, which has different reactivities for the different network topologies. We plot the distance from the synchronous state $\| \delta \bar{\bX}(t) \|$ {(see Methods \ref{App:Isolating} and Supplementary Note 3)}  as a function of time. Although both networks eventually achieve synchronization, we see large peaks in the transient time evolution in the case of the more reactive network, while these are not seen in the case of the less-reactive topology. {Supplementary Note 3 provides further illustrations of the effects of reactivity on synchronization dynamics by showing how either the choice of the initial conditions or of the network topology affects the occurrence of initial surges in the norm of the motion transverse to the synchronization manifold.} 

We now provide {as a novel contribution of this paper} a precise definition of the transverse reactivity of a point on the attractor.
The 
transverse reactivity 
of the perturbations about $\bx_s$ on the synchronous solution is given by
\begin{equation} \label{eq:rtsimple}
    r(\bx_s) = e_1 \left(\frac{\bD\bF( {\bx}_s)+\bD\bF^\top( {\bx}_s)}{2} +\sigma \xi H \right),
\end{equation}
where operator {$e_1(\cdot)$ computes the largest eigenvalue, }$\bD\bF({\bx}_s)$ is the Jacobian of $\bF$ at ${\bx}_s$ and the quantity,
\begin{equation}
    \xi = e_1 \left( V^\top \frac{L + {L}^\top}{2} V \right)
\end{equation}
is often referred to as the algebraic connectivity of directed graphs \cite{wu2005algebraic}, as a generalization of the classical concept of algebraic connectivity for undirected graphs \cite{fiedler1973algebraic}.
The matrix $V \in \mathbb{R}^{N\times N-1}$ is an orthonormal basis for the null subspace of $[1 \ 1 \hdots 1] \in \mathbb{R}^{1 \times N}$, i.e., the matrix $V$ is any matrix with normal columns that are orthogonal to $[1 \ 1 \hdots 1]^\top \in \mathbb{R}^N$ and to each other.
See Methods \ref{App:Isolating} for detailed derivation {and discussion of the transverse reactivity $r(\bx_s)$.}

The mapping that associates to each point $\bx_s$ of the attractor $\mathcal{A}$ its transverse reactivity $r(\bx_s)$, defines the \textit{reactive characterization of the attractor $\mathcal{C}(\mathcal{A})$.}
 A sufficient condition for $\xi = Re(\lambda_2)$ is that the Laplacian matrix $L$ is normal.
It also follows that for all undirected 
networks, $\xi = \lambda_2$.

Supplementary Note 4 presents upper and lower bounds for $\xi$. 
In particular, we prove that $\xi \leq 0$ for minimally reactive networks \cite{nazerian2023single}, i.e., a class of networks for which the largest eigenvalue of the symmetric part of the Laplacian is zero.
These networks, also known as balanced networks, are such that the in-degree and the out-degrees of each node are the same.
Having a negative $\xi$ implies that when the coupling strength is increased, the transverse reactivity of the points on the attractor either decreases or remains constant. 

Supplementary Note 5 presents the values of the reactivity $r(\bx_s)$ over a few well-known chaotic attractors: Lorenz, Rossler, Chen, Forced Van Der Pol, the FitzHugh-Nagumo model, and the Hindmarsh–Rose model.

\color{black}



\subsection{Efficient Synchronization Dynamics}

\subsubsection{Reactivity-based coupling scheme} \label{sec:dynamics}


{As stated in the introduction, our goal is to achieve coupling-efficiency and energy-efficiency of the synchronization dynamics, by requiring} lower coupling strength on average and lower synchronization energy in comparison to the case of constant coupling.
To this end, we find that  a simple modification
to Eq.\ \eqref{eq:main},  in which the constant coupling strength $\sigma$ is replaced by a time-varying one $\sigma(t)$, can be extremely beneficial,
\begin{equation}
\label{eq:timevarying}
\dot{\bx}_i(t)=\bF(\bx_i(t))+\sigma (t)\sum_{j=1}^N L_{ij}H \bx_j(t), \quad i=1,\hdots ,N.
\end{equation}

We call $\bu_i(t)=\sigma (t)\sum_{j=1}^N L_{ij}H \bx_j(t)$ the synchronization input affecting node $i$ in Eq.\ \eqref{eq:timevarying} and 
\begin{equation} \label{eq:energy}
    \mathcal{E}= \dfrac{1}{N} \sum_{i=1}^N \dfrac{1}{t_f - t_0} \int_{t_0}^{t_f}  \| \bu_i(t) \| dt
\end{equation}
the synchronization energy, corresponding to a given choice of the coupling strength $\sigma(t)$ over the time interval $[t_0,t_f]$, where $t_0$ and $t_f>t_0$ are some preassigned times. Here, $\| \cdot \|$ is the 2-norm. {Our definition of synchronization energy is derived from signal processing 
where the power of the scalar signal $s(t)$ is equal to $s(t)^2$ and the energy of the signal is equal to the integral of the power over time, $ \int_{-\infty}^{+\infty} s(t)^2 dt$, see e.g. \cite{Rihaczek1968Signal}. We acknowledge here that based on the particular selection of the individual systems and of the node-to-node coupling function, one may consider other definitions of the synchronization energy that are system-specific.} 

{
Our work applies to both cases that the MSF is negative in an unbounded range or bounded range of its argument \cite{Report}. When the range is unbounded (often referred to as Class II of the MSF),
as we increase $\sigma$ from zero, there is only one transition from asynchrony to synchrony (A $\rightarrow$ S) at the critical value $\sigma^{A \rightarrow S}$. When the range is bounded (Class III of the MSF), as we increase $\sigma$ from zero, first there is a transition from asynchrony to synchrony (A $\rightarrow$ S) at the critical value  $\sigma^{A \rightarrow S}$ followed by another transition from synchrony to asynchrony (S $\rightarrow$ A) at the critical value $\sigma^{S \rightarrow A}>\sigma^{A \rightarrow S}$.}

First, we consider the case of a transition from asynchrony to synchrony (A $\rightarrow $ S transition), 
for which the condition for stability of the synchronous solution is that $\sigma>\sigma^{A \rightarrow S}$ (the latter is a function of $\lambda_2$).
We proceed under the assumption that for a given choice of $\bF$, $H$, $L$, and constant coupling strength $\bar{\sigma}<\sigma^{A \rightarrow S}$, the coupled oscillators in Eq.\,\eqref{eq:timevarying} will not synchronize. 
We aim to find a time-varying coupling strength $\sigma (t)$ such that a) the average coupling strength is
    $1/T\int_0^T \sigma (t) dt = \bar{\sigma}$,
where $T$ is the total time, and b) the coupled dynamical systems in Eq.\,\eqref{eq:timevarying} synchronize.
We propose the following simple strategy which we call `coupling when needed' (CWN),
\begin{equation} \label{eq:sigmat}
    \sigma (t) = \begin{cases}
    \Bar{\sigma} \dfrac{1-\gamma (1 - \tau)}{\tau}, \quad & r(\bar{\bx}(t)) > \beta \\
    \\
    \Bar{\sigma}\gamma , & r(\bar{\bx}(t)) \leq \beta \\
\end{cases}
\end{equation}
where $\gamma$ and $\beta$ are tunable parameters such that $0 \leq \gamma \leq 1$ and  $\beta_{\min} < \beta < \beta_{\max}$.
The average solution at time $t$ is $\bar{\bx} (t) = \frac{1}{N} \sum_{i=1}^N \bx_i(t)$, and $\beta_{\min} = \min_{\bx_s \in \mathcal{A}} r(\bx_s)$ and $\beta_{\max} = \max_{\bx_s \in \mathcal{A}} r(\bx_s)$. 
The transverse reactivity $r(\bar{\bx}(t))$ is evaluated at $\bar{\bx}(t)$ using Eq.\,\eqref{eq:rtsimple} with $\sigma = \bar{\sigma}$.
Here, the parameter $0 < \tau < 1$ is the fraction of the times when $r(\bar{\bx}(t)) > \beta$ and is a function of $\beta$.
A good approximation for $\tau$ may be calculated beforehand using a long enough pre-recorded synchronous solution $\bs(t)$, Eq.\,\eqref{eq:xss}, as
    $\tau = 1/2 + 1/2\Big\langle \sign \Big( r(\bs(t) ) - \beta \Big)  \Big\rangle_{t}$.
{
Here, $\langle \cdot \rangle_t$ denotes the time average over the interval $t$. The sign function returns $1, 0$, or $-1$ when the argument inside is positive, zero, or negative, respectively.
}
As long as the initial conditions of the connected systems are close to the synchronous solution, the above approximation of $\tau$ is sufficiently close to the actual probability that $r(\bar{\bx} (t)) > \beta$.

If $\gamma = 1$, then $\sigma(t) = \bar{\sigma}, \ \forall t$, so the time-varying coupling strategy simplifies to the constant coupling. {Otherwise, the CWN strategy returns a stronger (weaker) coupling strength $\sigma$ when the transverse reactivity is larger (lower) than the threshold $\beta$.} If $\gamma=0$, the CWN strategy becomes on-off, similar to the work of Refs.\ \cite{belykh2004blinking,so2008synchronization,chen2009reaching,buscarino2017synchronization,PARASTESH2019Synchronizability}. A comparison between our work and these references is found in Supplementary Note 6, which shows a strong advantage of our CWN approach.

We now discuss the other case of a transition from synchrony to asynchrony (S $\rightarrow$ A transition), for which the condition for stability of the synchronous solution is that $\sigma<\sigma^{S \rightarrow A}$ (the latter is a function of $\lambda_N$).
We consider that $\bar{\sigma}$ is greater than the critical coupling $\sigma^{S \rightarrow A}$ predicted by the MSF analysis.
Hence, the system of our interest in Eq.\,\eqref{eq:timevarying} will not synchronize if $\sigma (t) = \bar{\sigma}$.
Our CWN strategy for the case of an S $\rightarrow$ A transition is,
\begin{equation} \label{eq:sigmat2}
    \sigma (t) = \begin{cases}
    \Bar{\sigma}\alpha, \quad & r(\bar{\bx}(t)) > \beta \\
    \\
    \Bar{\sigma} \dfrac{1-\tau \alpha}{1-\tau} , & r(\bar{\bx}(t)) \leq \beta \\
\end{cases}
\end{equation}
where $0 \leq \alpha \leq 1$ and $\beta_{\min} < \beta < \beta_{\max}$ are tunable parameters. 
{Similar to the CWN startegy for A $\rightarrow$ S, $\alpha = 0$ ($\alpha = 1$) corresponds to the constant coupling (the on-off coupling.)}
See Methods \ref{sec:cwn} for detailed derivations of the CWN strategies for both A $\rightarrow$ S and S $\rightarrow$ A transitions.

Supplementary Note 7 shows how this framework can be applied to linear consensus dynamics.  

We note here that in the presence of noise or small parametric mismatches, approximate synchronization of the set of Eqs.\ \eqref{eq:main} is still possible, but large desynchronization bursts known as bubbles may occur \cite{ashwin1994bubbling}. While linear stability analysis does not predict these bubbles, we show in Supplementary Note 8 that the transverse reactivity is able to explain them and that they can be eliminated using our coupling scheme.

\subsubsection{Relation to prior results}
In this section we briefly summarize a few of the major results from the area of synchronization in time-varying networks in order to place our own results in context; A thorough review of such results can be found in Ref. \cite{ghosh2022synchronized}. Of primary importance is the Master Stability Function (MSF): The MSF is the maximum transverse Lyapunov as a function of the coupling strength $\sigma$ and the eigenvalues $\lambda_i$ of the network Laplacian matrix (often denoted $\Lambda_{max}(\sigma\lambda_i)$); therefore, it can be used to identify intervals of the coupling strength within which the synchronous solution is asymptotically stable for a given network topology \cite{Pe:Ca}.

For networks in which the coupling strength varies on a time scale much faster than the node dynamics, the stability of synchronization is determined by the MSF of the average coupling strength $\bar{\sigma}$ and the eigenvalues of the network Laplacian (i.e., $\Lambda_{max}(\bar{\sigma}\lambda_i)$ \cite{stilwell2006sufficient}. Indeed, this is why in the following we show the synchronization error as a function of the average coupling strength. As the following sections demonstrate, our CWN strategy allows for synchronization to occur for values of $\bar{\sigma}$ for which it would not be possible with fast switching. 

For networks that vary smoothly in time (independent of time scale) in a state-independent way and such that the adjacency matrices commute with each other, the stability of synchronization is determined by the following condition: For each eigenvector of the (time-varying) adjacency matrix associated with perturbations transverse to the synchronization manifold, the associated time averaged maximal Lyapunov exponent of the variational equation must be negative (i.e., $S_i=\lim_{T\to\infty}\frac{1}{T}\int_0^T\Lambda_{max}(\sigma(t)\lambda_i(t))dt<0$ for all $i$ that correspond to transverse perturbations) \cite{boccaletti2006synchronization}. While this result implies that it may be possible for a network with time-varying coupling strength to synchronize when it would not synchronize with a constant coupling strength of $\bar{\sigma}$, to our knowledge such a demonstration, even anecdotally, has never been achieved.

In this work, we present the first demonstration that through well-designed time-varying coupling strength synchronization can be obtained over intervals of the average coupling strength that are significantly broader than those predicted by the traditional MSF analysis. Further, we provide a systematic way, based on the transverse reactivity, to achieve this reduction in average coupling strength, leading to substantial gains in the efficiency of synchronization.

\color{black}

\subsubsection{Examples}

\begin{figure*}
    \centering
    \includegraphics[width=0.85\textwidth]{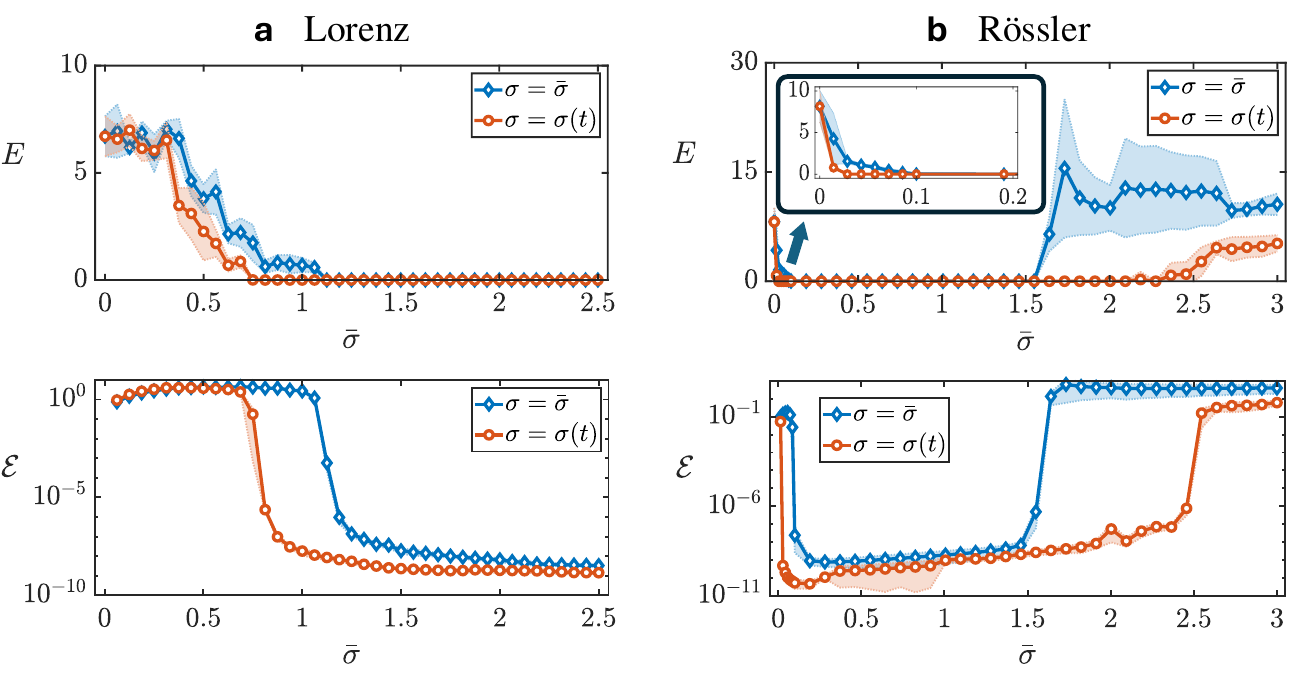}
    \caption{
    \textbf{The synchronization error $E$ (top) and the synchronization energy $\mathcal{E}$ (bottom) versus the average coupling strength $\Bar{\sigma}$.} Panel \textbf{a} shows the case of Lorenz systems. The parameters in Eq.\,\eqref{eq:sigmat} for $\sigma = \sigma (t)$ 
    $\beta = 0.5$ and $\gamma = 0.16$. Panel \textbf{b} shows the case of R{\"o}ssler oscillators. The parameters for $\sigma = \sigma (t)$ in Eqs.\,\eqref{eq:sigmat} and \eqref{eq:sigmat2} are $\beta = 0.2$ and $\gamma = \alpha = 0.01$. For $0 \leq \bar{\sigma} < 0.3$, we use Eq.\,\eqref{eq:sigmat}, and for $0.3 \leq \bar{\sigma} \leq 3$, we use Eq.\,\eqref{eq:sigmat2}.
    The data for both panels are averaged over 20 realizations initiated from randomly chosen initial conditions. The shaded backgrounds show the standard deviation of the plotted data.
    }
    \label{fig:sigmabar_compare}
\end{figure*}

We show the effectiveness of the CWN strategy in Eqs.\,\eqref{eq:sigmat} and \eqref{eq:sigmat2} through examples with coupled Lorenz oscillators and R{\"o}ssler oscillators.
 Other examples of different local dynamics such as the Hindmarsh-Rose neuron model, the FitzHugh-Nagumo neuron model, and the forced Van der Pol oscillator are presented in Supplementary Note 9. 

We define the synchronization error,
\begin{equation}
\label{eq:syncerror}
    E = \left< \frac{1}{N} \sum_{i=1}^N \| (\bx_i(t) - \Bar{\bx}(t)) \|  \right>_t.
\end{equation}
Here, $< \cdot >_t$ returns the average over the time interval $t \in [0.9 T \ \ T]$, and we set $T = 2000 \ s$.
The initial conditions for the oscillators are chosen randomly 
in a small neighborhood of the synchronous solution.
{
For the details of the dynamical function, coupling matrices, and the Laplacian matrix of this example, see Methods \ref{sec:exmethods}.
}

Figure \ref{fig:sigmabar_compare} shows the synchronization error for the oscillators when the coupling strategy in Eq.\,\eqref{eq:timevarying} is
\begin{enumerate}
    \item constant coupling $\sigma(t) = \Bar{\sigma}$
    \item time-varying coupling $\sigma(t)$ in either Eqs. \eqref{eq:sigmat} or \eqref{eq:sigmat2}.
\end{enumerate}

In the case of an A $\rightarrow$ S transition (S $\rightarrow$ A transition),
$\sigma = \sigma (t)$ from Eq. \eqref{eq:sigmat} (Eq. \eqref{eq:sigmat2}) is used.
Figure\,\ref{fig:sigmabar_compare} \textbf{a} shows the synchronization error $E$ as the average coupling strength $\bar{\sigma}$ is varied for a network of Lorenz oscillators.
Here, we set $\gamma = 0.16$ and $\beta = 0.5$ in Eq.\,\eqref{eq:sigmat}. {As an illustrative example, the reactive characterization of the Lorenz attractor with $\sigma\xi=-1$ is provided in Supplementary Fig. 2. For the A$\to$ S transition, the coupling is increased when $r(\bar{\bx} (t)) > \beta$.}

The time-varying coupling strategy $\sigma (t)$ successfully synchronizes the network with 
$\Bar{\sigma} = 0.75$ while the constant coupling strategy requires at least $\Bar{\sigma} = 1.12$.
This corresponds to a 
$33\%$ reduction of the critical average coupling.
The lower panel of Fig.\ \ref{fig:sigmabar_compare} \textbf{a} also shows that our proposed strategy corresponds to a substantial reduction in energy expenditure compared to the constant coupling strength case.
We conclude that our proposed strategy is capable of achieving both i) a reduction in the average coupling strength $\bar{\sigma}$ and ii) a reduction in the synchronization energy $\mathcal{E}$. 
See Supplementary Note 11 for an in-depth discussion on the energy efficiency of the strategy and how this scales with the average coupling strength $\bar{\sigma}$ for the case of connected Lorenz oscillators.

We now consider the case of R{\"o}ssler oscillators coupled in the $x$ variable and study separately the two transitions that are seen as the average coupling strength is increased: the A $\rightarrow$ S transition followed by the S $\rightarrow$ A transition.  {As an illustrative example, the reactive characterization of the R{\"o}ssler attractor with $\sigma\xi=-1$ is provided in Supplementary Fig. 2. For the A$\to$ S (S$\to$ A) transition, the coupling is increased (decreased) when $r(\bar{\bx} (t)) > \beta$.}
In the case of the A $\rightarrow$ S transition,  we set $\beta = 0.2$ and $\gamma = 0.01$ in Eq.\,\eqref{eq:sigmat} and in the case of the S $\rightarrow$ A transition, we set $\beta = 0.2$ and $\alpha = 0.01$ in Eq.\,\eqref{eq:sigmat2}.
Figure \ref{fig:sigmabar_compare} \textbf{b} (top) demonstrates a decrease of about 70\% of the critical average coupling when the time-varying coupling strategy is implemented for the A $\rightarrow$ S transition and an increase of about 70\% of the critical average coupling for the S $\rightarrow$ A transition.
Figure \ref{fig:sigmabar_compare} \textbf{b} (bottom) shows that by the use of time-varying coupling, the synchronization energy $\mathcal{E}$ is also reduced in comparison to constant coupling, which is seen over the entire range of $\Bar{\sigma}$ plotted in the figure.
We thus conclude that the time-varying coupling strategies in Eq.\,\eqref{eq:sigmat} and \eqref{eq:sigmat2} can be implemented successfully to significantly expand the range of the coupling strength in which the network synchronizes and to also reduce the synchronization energy $\mathcal{E}$.

{
In Supplementary Note 11, we have performed a synchronization energy comparison for the synchronization of coupled Lorenz oscillators from simulation data from Fig.\,\ref{fig:sigmabar_compare} \textbf{a}.
We did not see a significant difference in the synchronization error when the energy for the constant coupling and the CWN were the same.
However, for the same $\bar{\sigma}$, we saw that both energy and the synchronization error were higher in the case of constant coupling in comparison to the CWN.
}



We now study the effects of varying the two parameters $\beta$ and $\gamma$ in the synchronization strategy of Eq.\,\eqref{eq:sigmat} for the same system of coupled Lorenz oscillators studied in Fig.\,\ref{fig:sigmabar_compare}\textbf{a}. 
The MSF threshold for synchronization is $\sigma^{A \rightarrow S} Re(\lambda_2) \approx -2.3$ as reported in \cite{Pecora2009}.
Here, we wish to see how much smaller we can make $\bar{\sigma} Re(\lambda_2)$ than the MSF threshold and still observe synchronization.
To this end, we vary {$\gamma$ and $\beta$} in Eq.\,\eqref{eq:sigmat} and find the smallest 
\begin{equation*}
    \text{\% MSF threshold} \ A \rightarrow S = 100 {\dfrac{\bar{\sigma}}{\sigma^{A \rightarrow S} }.}
\end{equation*}
Figure\,\ref{fig:sigmabar} shows the \% MSF threshold $A \rightarrow S$ as $\gamma$ and $\beta$ are varied.
We see that the switching law in Eq.\,\eqref{eq:sigmat} can successfully synchronize the system for an average value of the coupling as low as 1\% of the critical coupling strength corresponding to the MSF threshold. 

In Supplementary Note 12, 
{we provide a similar example to what shown in Fig.\,\ref{fig:sigmabar} for the CWN strategy for S $\rightarrow$ A in Eq.\,\eqref{eq:sigmat2}.}
This strategy is applied to a network of coupled Lorenz oscillators and it is shown that the \% MSF threshold for an S $\rightarrow$ A transition can be increased up to five folds (530 \%).
This significant increase in the upper bound on $\bar{\sigma}$ demonstrates the effectiveness of the time-varying coupling strength in the case of an S $\rightarrow$ A transition.

{
As a final numerical example, we demonstrate that the synchronization scheme presented in Ref. \cite{Cavalcante2013Predictability} for the control of extreme events called dragon kings is a special case of our reactivity-based coupling scheme in Supplementary Note 8.} 

\begin{figure}
    \centering
\includegraphics[width=0.8\linewidth]{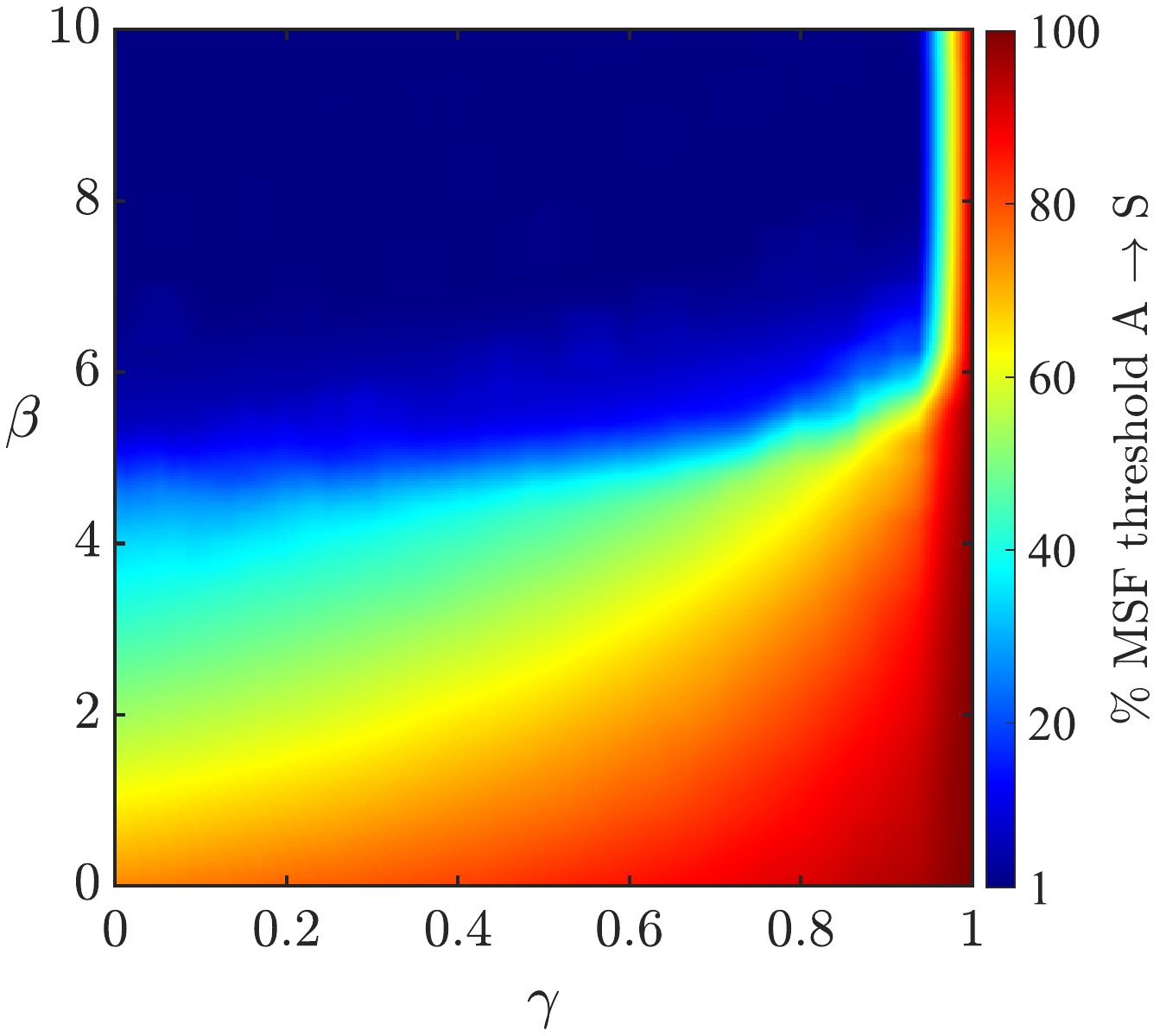}
\caption{\textbf{Effect of the parameter settings on the performance of the CWN strategy.} \% MSF threshold as the parameters $\gamma$ and $\beta$ in Eq.\,\eqref{eq:sigmat} are varied. {The details of these coupled Lorenz systems are presented in Methods \ref{sec:exmethods}.}
    }
    \label{fig:sigmabar}
\end{figure}

\subsubsection{Application to networks of opto-electronic oscillators}

We have demonstrated the efficacy of our time-varying coupling scheme in numerical simulations; however, networks in the real world are composed of non-identical oscillators and are subject to noise. Additionally, our coupling scheme relies on a model, which is bound to be imperfect. In this section, we test our reactivity-based coupling scheme on a network of two bi-directionally coupled, chaotic opto-electronic oscillators, and we find that our coupling scheme is robust in an experimental network.

The type of opto-electronic oscillator used here consists of a nonlinear, time-delayed feedback loop. These types of opto-electronic oscillators have found applications in areas such as communications \cite{argyris2005chaos}, microwave waveform generation \cite{maleki2011optoelectronic,hao2018breaking}, and photonic machine learning \cite{larger2012photonic}. A review of these devices can be found in Ref. \cite{chembo2019optoelectronic}.

A complete description of the opto-electronic oscillator experimental setup and coupling scheme is provided in Supplementary Note 13. A model for the dynamics of our opto-electronic network has been developed in previous work \cite{murphy2010complex}:
\begin{equation} \label{eq:OEO}
\begin{split}
    T\frac{d}{dt}x_i(t)=-x_i(t) + \beta_{fb}\cos^2(x_i(t-\tau_D)+\phi_0) \\ + \sigma(t) \sum_{j=1}^2L_{ij}\cos^2(x_j(t-\tau_D)+\phi_0),
\end{split}
\end{equation}
where $T$ is the low pass filter characteristic time, $\beta_{fb}$ is the round trip gain, $\sigma$ is the coupling strength, $L$ is the Laplacian coupling matrix, and $\phi_0 = \pi/4$. In this work, $L_{ij}$ = 1 for $i\neq j$ and $L_{ij}=-1$ for $i=j$. While opto-electronic oscillators can display a wide variety of dynamics \cite{murphy2010complex,chembo2019optoelectronic}, we tune our opto-electronic oscillators such that an uncoupled oscillator displays high dimensional chaotic dynamics by selecting $\beta_{fb}=4.0$, $\tau=500\mu s$, and $T=15.9\mu s$.


\begin{figure}
    \centering
    \includegraphics[width=0.8\linewidth]{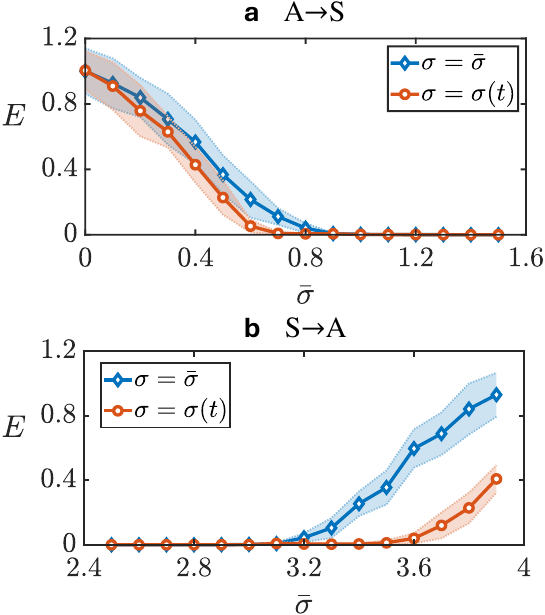} 
    \caption{
    \textbf{Demonstration of reactivity-based coupling scheme on an experimental network of two coupled opto-electronic oscillators.} The synchronization error $E$ is plotted against the average coupling strength $\Bar{\sigma}$ for the transitions from asynchrony to synchrony in panel \textbf{a} and synchrony to asynchrony in panel \textbf{b}. The blue (red) red curves show the constant (time-varying) coupling strategies. The shading shows the standard deviation.}
    \label{fig:expResults}
\end{figure}

First, to establish a baseline, we keep the coupling strength constant $\sigma=\bar{\sigma}$ and record the voltage applied to the modulator for each oscillator. The synchronization error between the two oscillators is shown in blue in Fig. \ref{fig:expResults}. Next, we implement the time-varying coupling scheme described in Sec.\,\ref{sec:dynamics} with $\beta=0.34$.
Although {the opto-electronic model in} Eq. \eqref{eq:OEO} is not in the standard form as Eq.\,\eqref{eq:main}, 
the dynamics of the transverse perturbations is very similar 
{the transverse dynamics corresponding to the coupled systems in Eq.\,\eqref{eq:main}.}
(see Supplementary Note 13).
The synchronization error with the time-varying coupling scheme is shown in red in Fig. \ref{fig:expResults}. In both cases, the scans over $\bar{\sigma}$ were performed ten times, and the shaded background shows the standard deviation of the measured synchronization errors. 
One can see that the minimum $\bar{\sigma}$ for A $\rightarrow$ S is reduced from 0.9 to 0.7 {($\%22$ reduction)} and the maximum $\bar{\sigma}$ for S $\rightarrow$ A is increased from 3.2 to 3.5 {($\%8.6$ increase.)} The results of the energy efficiency are presented in Supplementary Note 13.

We note that the computation of the reactivity relies on the model (Eq. \ref{eq:OEO}) and assumes that the oscillators are identical. 
These experimental results conclusively demonstrate that the time-varying coupling schemes presented in Sec.\,\ref{sec:dynamics} are robust to the imperfect model parameter estimations, non-identical oscillators, and noise that are inherently present in this experiment and in all real-world applications, and that our coupling strategy can be successfully applied to time-delayed systems.


\subsection{Network Syncreactivity} \label{sec:syncreactivity}

An important question is how 
the particular choice of the network topology affects the reactive characterization of the attractor and what we have discussed so far.
We proceed under the assumption that the particular choice of $\bF$ and $H$ corresponds to a master stability function that is negative in an unbounded range of its argument {(Class II MSF.)}
The other case in which the range is bounded {(Class III MSF)} is discussed in Supplementary Note 14.
We now want to compare two different network typologies 
in terms of the transverse reactivity $r(\bx_s)$, {which depends on $p = \sigma \xi$.}
However, for a proper comparison, it is required to pick $\sigma$  such that the long-term stability is the same for both networks. 
Given two network topologies with Laplacian matrices $L_A$ and $L_B$,
 we fix the coupling strength for each Laplacian matrix such that the long-term stability is the same,
that is,
    $\sigma_A Re(\lambda_2^A) = \sigma_B Re(\lambda_2^B) = a < 0$,
where $Re(\lambda_2^A)$ and $Re(\lambda_2^B)$ are the real part of the second eigenvalue of the Laplacian matrices $L_A$ and $L_B$, respectively.
Now, we would like to see if $\sigma_A \xi_A / \sigma_B \xi_B$ is less, equal, or larger than 1 where $\xi_A$  ($\xi_B$) is the algebraic connectivity of network A (network B), respectively.
From {section \ref{App:Isolating}, Property (ii)} for $r(\bx_s)$,
we know that $r(\bx_s)$ 
is a non-decreasing function of $\sigma \xi$.
Hence, it is higher for the Laplacian matrix $L_A$ than for the Laplacian matrix $L_B$ if $\sigma_A {\xi}_A>\sigma_B {\xi}_B$, or equivalently if the following condition is satisfied, $\xi_A/Re(\lambda_2^A) < \xi_B/Re(\lambda_2^B)$.

With this in mind, we introduce the  network syncreactivity index, 
\begin{equation} \label{eq:Xi}
    \Xi=1-\frac{\xi}{ Re(\lambda_2)},
\end{equation}
$\Xi \geq 0$ (see {section \ref{App:Isolating},} Property (i)),
\color{black}
 and note this is purely a topological measure of the network structure and reflects how reactive that network topology is.
If a network is connected and normal, then $\xi = Re(\lambda_2)$ and $\Xi=0$. 
{Note that normality is only a sufficient condition for $\Xi = 0$, not a necessary condition. 
 For example, the directed outward star, Network I in Fig.\,\ref{fig:main}\textbf{b}, has a non-normal Laplacian matrix but its index $\Xi = 0$.}
We emphasize that
 the network syncreactivity $\Xi$ is a single parameter of the network topology which is responsible for increasing/decreasing the reactive characterization of the attractor $\mathcal{C}(\mathcal{A})$. In particular, if for two networks A and B, {$\Xi^A > \Xi^B$, then $r^A(\bx_s) \geq r^B(\bx_s)$ for all  $\bx_s \in \mathcal{A}$.}

{
In Supplementary Note 15, we have investigated the effects of the syncreactivity $\Xi$ over the dynamics and} have seen that networks with higher $\Xi$ are more prone to the
 occurrence of bubbling \cite{ashwin1994bubbling}, both in terms of the number of bubbling events and of their size.

Supplementary Note 16 studies the syncreactivity $\Xi$ for two classes of synthetic directed unweighted networks: (i) Erd{\"o}s-R{\'e}nyi graphs, and (ii) scale-free graphs. We see that the syncreactivity $\Xi$ has an inverse relationship with the number of nodes, the connectivity probability of Erd{\"o}s-R{\'e}nyi networks, and the homogeneity of the degree distribution for scale-free networks.

\subsubsection{Effect of Syncreactivity Index $\Xi$ on the CWN Strategy}

In this subsection, we compare the performance of the CWN strategy for the cases of two network topologies characterized by different syncreactivity $\Xi$, namely a directed chain network and a directed star network, with $N=10$ nodes. The two $10 \times 10$ Laplacian matrices for these two networks are,
\begin{subequations} \label{eq:starchain}
    \begin{align}
        L_c & = \begin{bmatrix}
        0 & 0 & 0 & \hdots & 0 & 0 \\ 
        1 & -1 & 0 & \hdots & 0 & 0 \\ 
        0 & 1 & -1 & \hdots & 0 & 0 \\ 
        \vdots & \vdots & \vdots & \ddots & \vdots & \vdots \\ 
        0 & 0 & 0 & \hdots & -1 & 0 \\ 
        0 & 0 & 0 & \hdots & 1 & -1 \\ 
        \end{bmatrix}, \label{eq:chain} \\ 
        L_s & = \begin{bmatrix}
        0 & 0 & 0 & \hdots & 0 & 0 \\ 
        1 & -1 & 0 & \hdots & 0 & 0 \\ 
        1 & 0 & -1 & \hdots & 0 & 0 \\ 
        \vdots & \vdots & \vdots & \ddots & \vdots & \vdots \\ 
        1 & 0 & 0 & \hdots & -1 & 0 \\ 
        1 & 0 & 0 & \hdots & 0 & -1 \\ 
        \end{bmatrix},
    \end{align}
\end{subequations}
where the subscript $c$ ($s$) indicates chain (star.) From the lower triangular structure of the two Laplacian matrices $L_c$ and $L_s$, we see that their spectrum is the same, i.e., both Laplacian matrices have one $0$ eigenvalue and all the other eigenvalues are equal to  $-1$. There is however a difference in the syncreactivity $\Xi$ which results from the difference in the algebraic connectivity $\xi$. For the chain, $\xi_c = 0.1536$ and $\Xi_c = 1.1536$ while for the star, $\xi_s = -1$ and $\Xi_s = 0$. We note that both star and chain topologies have non-normal Laplacian matrices.
We use the same Lorenz settings as in Section IIB. 

Fig. \ref{fig:chainstar} shows the  \% MSF threshold A $\rightarrow$ S as the parameters $\beta$ and $\gamma$ are varied in  \eqref{eq:sigmat}. Panel \textbf{a} (\textbf{b}) contains the results for the chain (star) network topology. 
It is clear from the figures that for a fixed pair of the parameters $(\gamma , \beta)$, the performance of the CWN strategy is equal or worse for the coupled systems with the chain topology. 
Note that the best possible performance appears to be the same for both coupled systems at 1\% of the MSF threshold.
The best performance has been achieved for high values of $\beta$ and low values of $\gamma$, which corresponds to an on-off coupling strategy where the switching threshold is a high value of the reactivity of the attractor.


\begin{figure}
    \centering
    \includegraphics[width=0.75\linewidth]{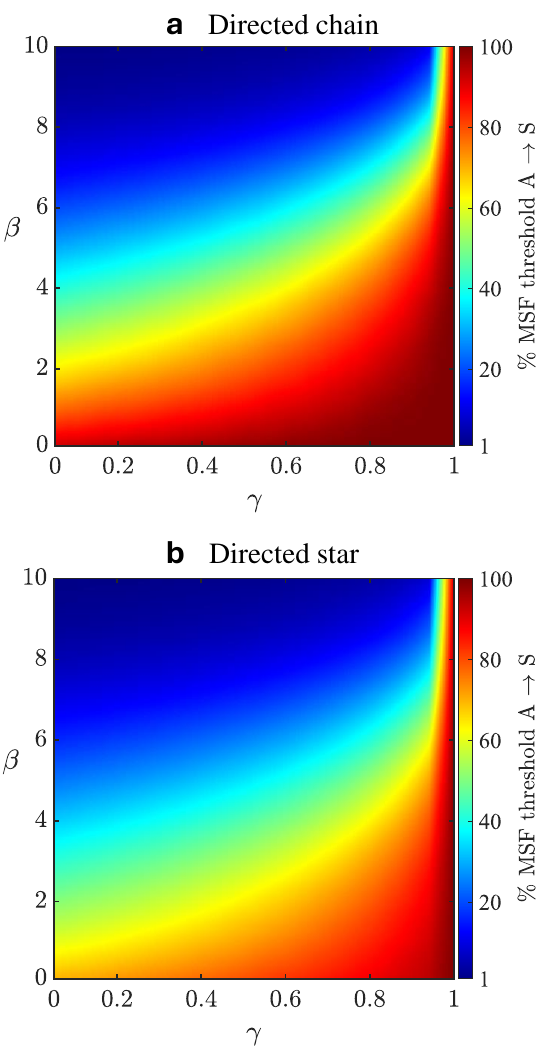}
    \caption{{
    \textbf{Effect of the syncreactivity $\Xi$ and the parameter settings on the performance of the CWN strategy.} We plot the \% MSF threshold as the parameters $\gamma$ and $\beta$ in Eq.\,\eqref{eq:sigmat} are varied. The two 10-node network topologies, {the chain in panel \textbf{a} and the star in panel \textbf{b}}, are described by the Laplacian matrices in Eq.\ \eqref{eq:starchain}, with the same spectrum but with different $\Xi$. For the chain, $\Xi_c = 1.1536$ and for the star, $\Xi_s = 0$.}}
    \label{fig:chainstar}
\end{figure}

\color{black}

\subsubsection{The syncreactivity of real networks}

Since the syncreactivity $\Xi$ is a parameter that solely depends on the structure of a network,
it is meaningful to study how this varies among different real networks from available data sets.
In what follows, for each network, we take the largest strongly connected component (LSCC) and evaluate $\Xi$ for its LSCC. 
Taking the LSCC of a network ensures that $Re(\lambda_2) \neq 0$ which guarantees synchronizability.

Figure\,\ref{fig:real} \textbf{a} plots the syncreactivity $\Xi$ of networks from different domains versus the network size $N$.
We see that on average, neural, trade, biological, and genetics networks are less reactive than social, metabolic, and file-sharing (Gnutella) networks.  
We also see that most of the more reactive networks have a larger number of nodes.
Figure\,\ref{fig:real} \textbf{b} is a plot of the syncreactivity $\Xi$ vs. the density, defined as the number of directed links in the network divided by $N^2$,  for the same set of real networks in Fig.\,\ref{fig:real} \textbf{a}. 
We see that the density correlates well with the syncreactivity, i.e., sparser (denser) networks have {higher (lower)} syncreactivity $\Xi$. 
Figure\,\ref{fig:real} \textbf{c} is a plot of the syncreactivity $\Xi$ vs. the synchronizability index $-Re(\lambda_2)$ (the larger $-Re(\lambda_2)$, the more synchronizable the network) for the same set of real networks in Fig.\ \ref{fig:real} \textbf{a}. 
Networks that are in the {bottom right} corner of the plot (e.g., neural) are more synchronizable and less reactive than those in the {top left} corner (e.g., metabolic)
and therefore they are more prone to synchronization both transiently and asymptotically. This is consistent with a conjecture that synchronization has been an evolutionary relevant principle in the formation of neural networks, but not in the formation of social, metabolic, and file-sharing networks \cite{Montgomery2023Evolutionary}.

\color{black}
In Supplementary Note 17, we have also plotted $\Xi$ {vs an index of non-normality} and other measures of synchronizability for directed networks such as 
the real-part eigenratio $Re(\lambda_2) / Re(\lambda_N)$ and the maximum imaginary part $I_{\max}$ among all eigenvalues of the Laplacian \cite{Ba:Pe02,Report}.

In Supplementary Note 18, 
\color{black}
further information on the real networks considered can be found.
{
In Supplementary Note 19, we have investigated the effect of the CWN strategy on the settling time of the synchronization dynamics. We have seen that the CWN reduces the settling time down to 8\% of the settling time with a constant coupling strategy. In Supplementary Note 20 we have considered the case of networks of phase oscillators.
}


To conclude, our analysis points out that there are at least two different purely topological indices of the ability of a network to synchronize: the synchronizability, characterizing the asymptotic synchronization dynamics, and the syncreactivity, characterizing the transient synchronization dynamics.  We argue here that when comparing different networks topologies in terms of their ability to synchronize,
both indices should be taken into account. 


\begin{figure}
    \centering
    \includegraphics[width=0.6\linewidth]{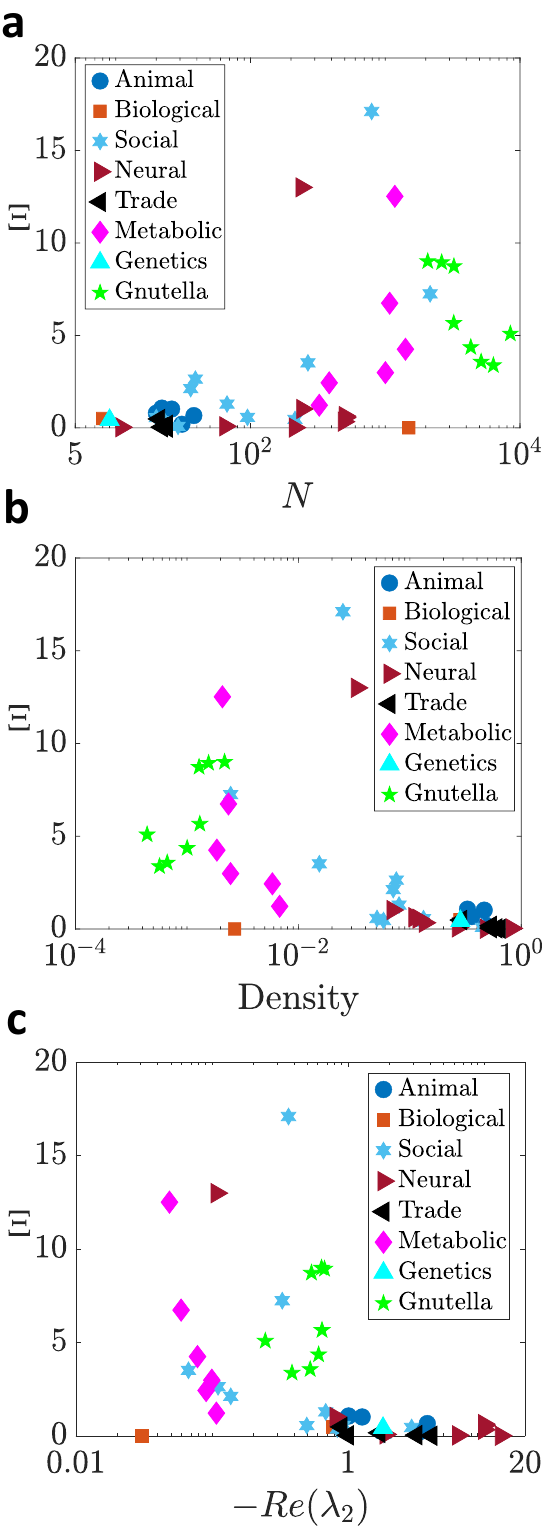}
    \caption{ \textbf{The syncreactivity of real-world complex networks.}
    The syncreactivity $\Xi$ vs \textbf{a} the size of the networks $N$, \textbf{b} the density, and \textbf{c} the negative of the real part of the second Laplacian eigenvalue $-Re(\lambda_2)$ for a collection of real-world networks from different domains.}
\label{fig:real}
\end{figure}

\section{Discussion}

Synchronization is a fundamental physical phenomenon that occurs in networks of coupled technological and biological systems. Much previous work has focused on the asymptotic stability of the synchronous solution, while this paper investigates the transient dynamics 
and explores the important question of the efficiency of the synchronization dynamics. By combining transient and asymptotic considerations, we achieve an exhaustive characterization of the synchronization dynamics of complex networks. 
This work advances the area of studies on synchronization of networks in more than one direction, as discussed below.

\textit{CWN synchronization strategy.} All oscillating systems 
move through regions of phase space that 
are different from one another: 
for example, in certain parts of an oscillation, synchronization may be possible for very little coupling or even for no coupling, while other parts may require strong coupling.
While the Lyapunov exponents provide average asymptotic measures of stability for a given attractor, {they fail at describing transient dynamical behavior. Our work supersedes the Lyapunov exponents analysis by considering a 
characterization of the reactivity of different regions of the synchronous attractor. This provides the motivation for exploring} new synchronization strategies for which the coupling strength is properly adjusted to different parts of oscillations (regions of the synchronous attractor.)  
 Our main result in this paper is the formulation of a synchronization strategy for networks of coupled oscillators that uses coupling only when needed: e.g., for the case that synchronization requires a large enough coupling strength (A $\rightarrow$ S transition), the coupling is increased when the transverse reactivity is large, and it is reduced otherwise. We showed the successful application of this strategy in simulations and experiments, and for a variety of different oscillators, including Lorenz, R{\"o}ssler, the forced Van der Pol, the Hindmarsh-Rose neuron model, the FitzHugh-Nagumo neuron model, and an opto-electronic oscillator experiment, and for different choices of the node-to-node coupling functions. We also showed that CWN provides a rigorous, general foundation for the control of extreme events such as dragon kings, which has previously been thought impossible \cite{Cavalcante2013Predictability}. 
We note that our proposed CWN strategy encompasses the traditional constant coupling and also the largely studied on-off coupling strategy. 
A gradual shift from the constant coupling strategy to on-off coupling is possible by controlling a scalar parameter between $0$ and $1$.
As a result, the CWN strategy is a powerful and versatile strategy to choose the coupling strength at each time.
The choice of the parameters is important in achieving the best efficiency of synchronization which is discussed next. 
Yet it is difficult to make a general assertion about what would be the optimal values of the parameters for all choices of dynamics, topologies, and coupling functions.

 \color{black}

\textit{Efficiency of the synchronization dynamics.}
A large part of the literature has focused on the conditions to ensure the stability of the synchronized state, while the important issue of the efficiency of the synchronization dynamics has so far received less attention. 
We investigate the issues of coupling-efficiency and energy-efficiency, which are relevant to both the biological world and technological applications.
We propose a synchronization strategy that achieves efficiency by only using coupling when needed.  {This has immediate benefits 
in terms of the actuators that can be used to achieve and maintain synchrony. In fact, 
both technological and biological systems 
are limited in the duration and overall intensity of the forces 
that they can exert and benefit from lower energy expenditures.} \textcolor{black}{Given the strong advantages we have observed in terms of both average coupling and energy expenditure, it 
appears likely that 
coupling and energy-efficient synchronization strategies 
may be implemented in the biological world. }

\textit{Enabling synchronization.}
Another motivation for this study is the observation that in several practical applications, the type of oscillators, the specific choice of the node-to-node coupling function and of the network topology cannot be changed. Hence, it is meaningful
to develop strategies to enable synchronization when it would not occur for a given type of oscillators, network topology, and node-to-node coupling. 
Our proposed synchronization strategy is exceptionally successful at synchronizing networks of coupled oscillators. The CWN method is particularly attractive because it enables energy-efficient synchronization such that the attractor of the network of synchronized oscillators is the same as the attractor of a single, uncoupled oscillator. This is in contrast to, e.g., synchronization induced by a common drive (including stochastic synchronization \cite{pikovsky1992statistics,herzel1995chaos}), in which the attractor of the synchronized network is qualitatively different from the attractor of a single, undriven oscillator due to the presence of the drive signal.
 By using our CWN strategy we were able to show a significant enlargement of the range of the average coupling strength over which synchronization arises.  
In particular,  in the case of an A $\rightarrow$ S transition (S $\rightarrow$ A transition) we could significantly reduce  (increase) the critical value of the average coupling strength over which synchronization could be established, sometimes by orders of magnitude. 
For example, in networks of Lorenz oscillators coupled in the second state variable, we achieved synchrony for an average value of the coupling strength as low as 1\% of the critical coupling strength predicted by the MSF analysis.

\textit{Network Syncreactivity.} We further introduced a new structural network property that characterizes the transient dynamics of networks towards synchronization, 
which we call network syncreactivity.  
 Several works have linked the reactivity to the non-normality of the dynamics' Jacobian.
It is known that systems characterized by a non-normal Jacobian are prone to transient effects, which may steer their long-term dynamics away from an equilibrium point, even when this is asymptotically stable \cite{trefethen1993hydrodynamic,Asllani2018Structure,asllani2018topological}. For equilibrium points, transient stability can be measured by the reactivity of the fixed point, which is defined as the initial rate of growth of a perturbation about the equilibrium point \cite{Neubert1997ALTERNATIVES,Duan2022Network}. 
{Although our work can be applied to both the cases of undirected and directed networks, it is particularly relevant to the latter, as these may have nonzero syncreactivity. }
We have found that the overall propensity of a network to synchronize can be fully described in terms of two topological scalar indices, synchronizability, and syncreactivity. 
An analysis of real complex networks from several domains has shown that typically neural networks have better transient and asymptotic synchronization properties than social, metabolic, and file-sharing networks. This is 
consistent with different evolutionary principles guiding the formation of networks from different domains.
We have also identified the density of connections to be a network topological property that well correlates with the syncreactivity while being distinct from previously introduced topological correlates \cite{Asllani2018Structure,Duan2022Network}.

{
\textit{Limitations and future directions.}
A limitation of our proposed CWN strategy is that the selection of the parameters $\beta$ and $\gamma$ requires ad-hoc tuning. {While optimization of these parameters can be nontrivial, in practice the CWN strategy is advantageous even when the parameters are not optimized. }
 {To see this, it suffices to look
at Fig. \ref{fig:sigmabar}, where all the points correspond to an
improvement in the threshold for synchronization, varying from a minimum of zero (dark red region) to 100 folds
(dark blue region.)} 
An important point of this paper is the connection between reactivity and the largest Lyapunov exponent of a system, where the former (the latter) measures the instantaneous (asymptotic) rate of growth of the state vector. Another related concept that we think deserves further investigation in relation to both reactivity and Lyapunov exponents is that of contraction \cite{lohmiller1998contraction}. } {Another point that is left for future investigation is the existence of a limit on how little average coupling can be spent and still achieve synchronization. For example, for the case of Lorenz systems presented in Fig.\ \ref{fig:sigmabar}, the smallest average coupling needed for synchronization seems to be close to 1\% of the amount needed for constant coupling. In order to characterize this limit, one would need to formulate a separate optimization problem for which the goal is to minimize the average coupling expenditure. }




\section{Methods}

\subsection{Isolating the dynamics transverse to the synchronous solution \label{App:Isolating}}

\color{black}

In order to study the stability of the dynamics of Eqs.\,\eqref{eq:main} about the synchronous solution $\bs(t)$, we linearized 
\eqref{eq:main}
about $\bs(t)$, thus obtaining,
\begin{equation} \label{eq:dx0}
\delta\dot{\bx}_i(t)=\bD\bF(\bs (t))\delta\bx_i(t)+\sigma \sum_{j = 1}^N L_{ij} H  \delta\bx_j(t),
\end{equation}
$i=1,...N$, where $\delta \bx_i(t) = \bx_i(t) - \bs (t)$ is a small perturbation and $\bD\bF(\bs (t))$ 
is the Jacobian evaluated at the synchronous solution at time $t$.
Equation\ \eqref{eq:dx0} is rewritten in the compact form as
\begin{equation} \label{eq:dx}
\delta\dot{\bX}(t)=[I_N \otimes \bD\bF(\bs (t))+\sigma  L \otimes H ] \delta\bX = Z(t) \delta\bX(t),
\end{equation}
where the $nN$-dimensional vector ${\delta\bX}^\top=  [\delta{\bx}_1^\top, \delta{\bx}_2^\top,...,\delta{\bx}_N^\top]$, $I_N$ is the $N$-dimensional identity matrix,
 and $\otimes$ denotes the Kronecker product. The first challenge, which does not arise in the study of equilibrium points, is that of decoupling the synchronous `parallel' motion from the `transverse' motion.



\color{black}

 The variational system \eqref{eq:dx} 
has a `parallel dynamics' along the direction spanned by the eigenvector $[\sqrt{N}^{-1},\sqrt{N}^{-1},...,\sqrt{N}^{-1}]^\top$ corresponding to the only zero eigenvalue $\lambda_1=0$ and a `transverse dynamics' in the subspace orthogonal to this eigenvector.  
We are especially interested in characterizing the transverse dynamics. In order to isolate this transverse dynamics, we construct an orthogonal matrix $\tilde{V}$ having its first column equal to the vector $[\sqrt{N}^{-1},\sqrt{N}^{-1}, \hdots ,\sqrt{N}^{-1}]^\top$. This can be done, for example, by using the Gram-Schmidt method. Then, we consider the similarity transformation $\tilde{L}=\tilde{V}^\top L \tilde{V}$. In the general case in which the matrix $L$ is not symmetric, the matrix $\tilde{L}$ has the following structure,
\begin{equation*}
  \tilde{L} = \left[\begin{array}{@{}c|cccc@{}}
    0 & \tilde{L}_{1,2} & \tilde{L}_{1,3} & \dots & \tilde{L}_{1,N} \\\hline
    0 & \tilde{L}_{2,2} & \tilde{L}_{2,3} & \dots & \tilde{L}_{2,N} \\
    0 & \vdots & \vdots & \ddots & \vdots \\
    0 & \tilde{L}_{N-1,2} & \tilde{L}_{N-1,3} & \dots & \tilde{L}_{N-1,N} \\
    0 & \tilde{L}_{N,2} & \tilde{L}_{N,3} & \dots & \tilde{L}_{N,N} 
  \end{array}\right],
\end{equation*}
where we call the $(N-1)$-dimensional block in the right-lower corner the reduced matrix $L^\perp$. 
Alternatively, one can retrieve $L^\perp$ by first removing the first column of the matrix $\tilde{V}$ to obtain $V$ and then
\begin{equation}
    L^\perp = V^\top L V.
\end{equation}
Note that by construction the matrix $L^\perp$ has all negative real-part eigenvalues. Applying the transformation $\tilde V$, we can then write down the equation for the time evolutions of the transverse motions corresponding to Eq.\eqref{eq:dx},
 \begin{equation} \label{deltaXX}
\delta\dot{\hat \bX}(t)=[I_{N-1} \otimes \bD\bF(\bs (t))+\sigma  L^\perp \otimes H ] \delta\hat{\bX} = \hat{Z}(\bs (t)) \delta\hat{\bX}(t).
\end{equation}

We define the 
transverse reactivity 
of the perturbations about $\bx_s$ on the synchronous solution
\begin{align} \label{eq:rt}
\begin{split}
    r(\bx_s) := &  \frac{1}{2} e_1 \left(\hat{Z}^\top(\bx_s)+ \hat{Z}(\bx_s) \right) \\
    = &  \frac{1}{2} e_1 \bigg( I_{N-1} \otimes [\bD\bF( {\bx}_s)+\bD\bF^\top( {\bx}_s)]  \\
    & \qquad + \sigma ({L^\perp}^\top + L^\perp) \otimes H \bigg).
\end{split}
\end{align}
We remark that the transverse reactivity $r(\bx_s)$ determines the reactivity associated with Eq.\,\eqref{deltaXX} at a particular point $\bx_s$ on the synchronous solution.
If $r(\bx_s) > 0$ ($r(\bx_s) < 0$), then the norm of transverse perturbations $\| \delta \hat{\bX} \|$ can (cannot) increase instantaneously.

We remark that through Eq.\ \eqref{eq:rtsimple}, the transverse reactivity depends on the parameter $p=\sigma \xi$, where $\sigma \geq 0$ is the coupling strength and $\xi$ is the algebraic connectivity.

In what follows, we will simplify Eq.\,\eqref{eq:rt} to obtain Eq.\,\eqref{eq:rtsimple}.
We write down the eigenvalue equation for the symmetric matrix $S_{L^\perp}=({L^\perp} +{L^\perp}^\top)/2$,
{
$S_{L^\perp} V_S=V_S Y$, where the columns of the orthogonal matrix $V_S$  are the eigenvectors of the matrix $S_{L^\perp}$
} 
and the matrix $Y$ is diagonal with the elements on the main diagonal equal to the eigenvalues of $S_{L^\perp}$. 
{The largest eigenvalue of the symmetric matrix $S_{L^\perp}$ is often referred to as the algebraic connectivity, and here, we denote it as $\xi = e_1(S_{L^\perp})$.}
Then, we can rewrite {$r(\bx_s)$} by pre-multiplying and post-multiplying Eq.\,\eqref{eq:rt} by 
{
$V_S^\top \otimes I$ and $V_S \otimes I$,
}
respectively, yielding,
\begin{equation} 
    r(\bx_s)= e_1(I \otimes [\bD\bF( {\bx}_s)+\bD\bF^\top( {\bx}_s)]/2+\sigma  Y \otimes H ).
\end{equation}
Because the matrix $Y$ is diagonal, then
\begin{equation} \label{sr}
 r(\bx_s)= \max_i \Bigl\{ e_1([\bD\bF( {\bx}_s)+\bD\bF^\top( {\bx}_s)]/2 +\sigma Y_{ii}  H) \Bigr\}.
\end{equation}
From Eq.\,\eqref{sr}, we also see that there are two distinct effects on the overall transverse reactivity, a baseline effect of the individual dynamics given by $[\bD \bF( {\bx}_s)+\bD \bF^\top ( {\bx}_s)]/2$ and an effect of the network topology given by $Y_{ii}$. 
The baseline effect depends on the particular choice of the function $\bF$ so that different choices of oscillators result in different baseline effects. 
In what follows we are particularly interested in the role of the network topology, so we focus on the largest eigenvalue of $S_{L^\perp}$.
Also, under the generic assumption that the matrices have simple spectra, one can show that (see \cite[Chapter 1.3.4]{tao2012topics} and \cite{horn1998eigenvalue})
\begin{equation*}
    \frac{d e_1([\bD\bF( {\bx}_s)+\bD\bF^\top( {\bx}_s)]/2 +\sigma Y_{ii} H)}{d Y_{ii}}={\pmb \pi}^\top \sigma H {\pmb \pi} \geq 0,
\end{equation*}
where ${\pmb \pi}$ is the Perron-Frobenius eigenvector (with entries all of the same sign) of the symmetric matrix $[\bD\bF( {\bx}_s)+\bD\bF^\top( {\bx}_s)]/2 +\sigma Y_{ii} H$.
We thus expect the maximum in Eq.\ \eqref{sr} to be always achieved for $i=i^*$ corresponding to {the algebraic connectivity $\xi$.}
Then, Eq.\,\eqref{eq:rt} for the  reactivity of the transverse motion is rewritten as
\begin{equation*} 
    r(\bx_s) = e_1 \left(\frac{\bD\bF( {\bx}_s)+\bD\bF^\top( {\bx}_s)}{2} +\sigma \xi H \right),  
\end{equation*}
which is the same as Eq.\,\eqref{eq:rtsimple}.

Next, we present some 
properties of the algebraic connectivity $\xi$ and of the transverse reactivity $r(\bx_s)$:
\begin{enumerate}[label=(\roman*)]
    \item $\xi \geq Re(\lambda_2)$, i.e., the algebraic connectivity $\xi$ is always greater than or equal to the real part of the second smallest eigenvalue of the Laplacian, $Re(\lambda_2)$. 
    \item For each point $\bx_s$, the transverse reactivity $r(\bx_s)$ (and so the reactive characterization of the attractor) 
   is a continuous monotonically non-decreasing function of the parameter $p=\sigma \xi$.
    \item {The transverse reactivity $r(\bx_s)$}
    is a continuous function of the synchronous solution $\bs (t)$ if
    the Jacobian $\bD \bF$
    is a continuous function of the synchronous solution.
\end{enumerate}

These properties are proved in the following sections of the Methods. 
From Property (ii) it follows that the transverse reactivity is a continuous monotonically non-decreasing function of $\xi$ for a fixed $\sigma$. 
    {For a fixed $\xi > 0$ ($\xi < 0$), the transverse reactivity is a continuous monotonically non-decreasing (non-increasing) function of $\sigma$.}

Based on 
Property (iii),
we can divide the attractor $\mathcal{A}$ into two distinct regions:
\begin{enumerate}
    \item The reactive region $\mathcal{R} = \{ \bx_s |  r(\bx_s) > 0, \ \bx_s \in \mathcal{A}\}$, and
    \item The non-reactive region $\mathcal{N} = \{ \bx_{s} |  r(\bx_{s}) \leq 0 , \ \bx_s \in \mathcal{A}\}$,
\end{enumerate}
where $\mathcal{R} \bigcap \mathcal{N} = \emptyset$, $\mathcal{R} \bigcup \mathcal{N} = \mathcal{A}$.

We note that for a given choice of the function $\bF$, the reactivity of these regions is a function of the coupling strength $\sigma$, of the algebraic connectivity $\xi$, and the node-to-node coupling matrix $H$.
Thus, if any of the aforementioned parameters change while the local dynamics $\bF$ is fixed, the reactive and non-reactive regions change too.
{
The ratio between the size of $\mathcal{R}$ and the size of $\mathcal{A}$ defines the critical probability $\mu$ of observing an increase in the norm of the transverse perturbation at the initial time. 
{For detailed definition of the critical probability $\mu$, see Methods Sec.\,\ref{sec:prob}.}
}

\color{black}


\subsection{Coupling when needed} \label{sec:cwn}

We aim to find a time-varying coupling strength $\sigma (t)$ such that a) the average coupling strength is
\begin{equation} \label{eq:int}
    \frac{1}{T}\int_0^T \sigma (t) dt = \bar{\sigma},
\end{equation}
where $T$ is the total time, and b) the coupled dynamical systems in Eq.\,\eqref{eq:timevarying} synchronize.
We propose the following simple strategy which we call `coupling when needed' (CWN),
\begin{equation} \label{eq:sigmat0}
    \sigma (t) = \begin{cases}
        \sigma_1, \quad & r(\bar{\bx}(t)) > \beta \\
        \sigma_2, & r(\bar{\bx}(t)) \leq \beta \\
    \end{cases}
\end{equation}
where  $\bar{\bx} (t) = \frac{1}{N} \sum_{i=1}^N \bx_i(t)$ is the average solution at time $t$, $\beta_{\min} < \beta < \beta_{\max}$ is a tunable parameter, between $\beta_{\min} = \min_{\bx_s \in \mathcal{A}} r(\bx_s)$ and $\beta_{\max} = \max_{\bx_s \in \mathcal{A}} r(\bx_s)$, and
\begin{equation}
    r(\bar{\bx}(t)) = e_1 \left(\frac{\bD\bF( \bar{\bx}(t))+\bD\bF^\top( \bar{\bx}(t))}{2} +\bar{\sigma} \xi H \right).
\end{equation}
Here, $\xi$ is the previously introduced algebraic connectivity of the Laplacian $L$.
We proceed to find $\sigma_1$ and $\sigma_2$ such that $\sigma_1 \geq \Bar{\sigma} \geq \sigma_2 \geq 0$ and Eq.\,\eqref{eq:int} is satisfied.

Without loss of generality, we can set $\sigma_1 = \Bar{\sigma} / \alpha$ and $\sigma_2 = \Bar{\sigma} \gamma$ where $0 < \alpha \leq 1$ and $0 \leq \gamma \leq 1$.
By enforcing the constraint in Eq.\,\eqref{eq:int}, we get
\begin{equation*} \label{eq:constraint}
    \frac{1}{T}\int_0^T \sigma (t) dt \approx \frac{1}{T} \left( \tau T \frac{\bar{\sigma}}{ \alpha} + (1-\tau)T\bar{\sigma} \gamma \right) = \Bar{\sigma}.
\end{equation*}
Here, the parameter $0 < \tau < 1$ is the fraction of the times when $r(\bar{\bx}(t)) > \beta$ and is a function of $\beta$.
After simplifications, we obtain $1/\alpha = (1-\gamma (1 - \tau)) / \tau$.
Thus, Eq.\,\eqref{eq:sigmat0} is rewritten as
\begin{equation} 
    \sigma (t) = \begin{cases}
    \Bar{\sigma} \dfrac{1-\gamma (1 - \tau)}{\tau}, \quad & r(\bar{\bx}(t)) > \beta \\
    \\
    \Bar{\sigma}\gamma , & r(\bar{\bx}(t)) \leq \beta \\
\end{cases}
\end{equation}
where $\gamma$ and $\beta$ are tunable parameters such that $0 \leq \gamma \leq 1$ and  $\beta_{\min} < \beta < \beta_{\max}$. 
If $\gamma = 1$, then $\sigma(t) = \bar{\sigma}, \ \forall t$, so the time-varying coupling strategy simplifies to the constant coupling. If $\gamma=0$ our strategy becomes on-off, similar to the work of Refs.\ \cite{belykh2004blinking,so2008synchronization,chen2009reaching,buscarino2017synchronization,PARASTESH2019Synchronizability}. 
A good approximation for $\tau$ may be calculated beforehand using a long enough pre-recorded synchronous solution $\bs(t)$, Eq.\,\eqref{eq:xss}, as
\begin{equation*}
    \tau = \frac{1}{2} + \frac{1}{2}\Big\langle \sign \Big( r(\bs(t) ) - \beta \Big)  \Big\rangle_{t}.
\end{equation*}
As long as the initial conditions of the connected systems are close to the synchronous solution, the above approximation of $\tau$ is sufficiently close to the actual probability that $r(\bar{\bx} (t)) > \beta$.

We now focus on the other case of a transition from synchrony to asynchrony (S $\rightarrow$ A transition), for which the condition for stability of the synchronous solution is that $\sigma<\sigma^{S \rightarrow A}$ (the latter is a function of $\lambda_N$).
We consider that $\bar{\sigma}$ is greater than the critical coupling $\sigma^{S \rightarrow A}$ predicted by the MSF analysis.
Hence, the system of our interest in Eq.\,\eqref{eq:timevarying} will not synchronize if $\sigma (t) = \bar{\sigma}$.

To synchronize the system under a state-dependent coupling strategy with an average value of $\bar{\sigma}$, we use the same coupling strategy in Eq.\,\eqref{eq:sigmat0} but for this case, we set $0 \leq \sigma_1 \leq \Bar{\sigma} \leq \sigma_2$.
Without loss of generality, we can take $\sigma_1 = \bar{\sigma} \alpha$ and $\sigma_2 = \bar{\sigma}/ \gamma$ where $0\leq \alpha \leq 1$ and $0<\gamma \leq 1$ are tunable parameters.
After enforcing the constraint in Eq.\,\eqref{eq:int}, we obtain
$1 / \gamma = (1 - \tau \alpha) / (1 - \tau)$,
where $\tau$ is the fraction of the times when $r(\Bar{\bx}(t)) > \beta$, as before.
Therefore, our CWN strategy for the case of an S $\rightarrow$ A transition is,
\begin{equation} 
    \sigma (t) = \begin{cases}
    \Bar{\sigma}\alpha, \quad & r(\bar{\bx}(t)) > \beta \\
    \\
    \Bar{\sigma} \dfrac{1-\tau \alpha}{1-\tau} , & r(\bar{\bx}(t)) \leq \beta \\
\end{cases}
\end{equation}
where $0 \leq \alpha \leq 1$ and $\beta_{\min} < \beta < \beta_{\max}$ are tunable parameters.

\subsection{Example details} \label{sec:exmethods}

The local dynamics $\bF$ and the coupling matrix $H$ for the case of Lorenz are 
\begin{equation} \label{eq:Lorenz}
    \bx = \begin{bmatrix}
        x \\ y \\ z
    \end{bmatrix}, \ \
    \bF(\bx) = \begin{bmatrix} 10(y-x) \\
    x(28 - z) - y \\
    xy - 2 z \end{bmatrix}, \ \  H = \begin{bmatrix}
        0 & 0 & 0 \\
        0 & 1 & 0 \\
        0 & 0 & 0 
    \end{bmatrix},
\end{equation}
which results in an unbounded range of the coupling strength for synchronization.
For the case of the R{\"o}ssler oscillator, we set
\begin{equation} \label{eq:Rossler}
    \bx = \begin{bmatrix}
        x \\ y \\ z
    \end{bmatrix}, \ \
    \bF(\bx) = \begin{bmatrix} -y-x \\
    x + 0.2 y  \\
    0.2 + (x-9)z \end{bmatrix}, \ \  H = \begin{bmatrix}
        1 & 0 & 0 \\
        0 & 0 & 0 \\
        0 & 0 & 0 
    \end{bmatrix},
\end{equation}
which results in a bounded range of the coupling strength that produces synchronization.
We randomly construct a directed unweighted graph, with Laplacian
\begin{equation} \label{eq:laplacian}
    L = \begin{bmatrix}
        -1 & 0 & 1 & 0 \\ 
        1 & -2 & 1 & 0 \\ 
        0 & 1 & -1 & 0 \\ 
        1 & 1 & 1 & -3 
    \end{bmatrix}.
\end{equation}

\color{black}

\subsection{Proof of Property (i) 
} \label{sec:proofxi}
Property (i)
follows from the fact that the largest eigenvalue of the symmetric part of a matrix is always greater than or equal to the largest real part eigenvalue of that matrix; therefore $\xi \geq Re(\lambda_2)$.
 The inequality is satisfied with the equal sign, i.e., $\xi = Re(\lambda_2)$, whenever the left and the right eigenvectors of $L^\perp$ are real and coincide. 
 The proof is complete. \hfill \qed

\subsection{Proof of Property (ii) 
\label{app:proof}}
We fix a point on the synchronous solution, $\bx_s \in \{\bs(t) \}$. Then, for an assigned Jacobian $\bD \bF (\bx_s)$ and  coupling matrix $H$, we look at the effects of varying $\sigma \xi$ on the eigenvalues of the matrix 
\begin{equation*}
    M = \frac{\bD\bF( {\bx}_s)+\bD\bF^\top( {\bx}_s)}{2} +\sigma \xi H.
\end{equation*}
As $\sigma \xi$ changes continuously, the entries of the matrix $M$ vary continuously as well.
It is well known that the eigenvalues of a matrix vary continuously with the entries of the matrix.
Therefore, $r(\bx_s)$ varies continuously with $\sigma \xi$.
Also, under the generic assumption that the matrices have simple spectra, one can show that (see \cite[Chapter 1.3.4]{tao2012topics})
\begin{equation*}
    \frac{d (r(\bx_s))}{d (\sigma \xi)} = \frac{d e_1(M)}{d (\sigma\xi)}={\pmb \pi}^\top  H {\pmb \pi} \geq 0,
\end{equation*}
where ${\pmb \pi}$ is the Perron-Frobenius eigenvector (with entries all of the same sign) of the symmetric matrix $M$.
Hence, $r(\bx_s)$ is a continuous monotonically non-decreasing function of $\sigma \xi$.
The proof is complete. \hfill \qed

\subsection{Proof of Property (iii) 
\label{sec:proofr}}
Consider the matrix
\begin{equation*}
    M = \frac{\bD\bF( {\bx}_s)+\bD\bF^\top( {\bx}_s)}{2} +\sigma \xi H
\end{equation*}
for a point on the attractor, $\bx_s \in \mathcal{A}$.
If we assume that $\bD\bF( {\bx}_s)$ is a continuous function of ${\bx}_s$, it follows the entries of $M$ are a continuous function of $\bx_s$.
Also, it is known that the eigenvalues of a matrix are continuous functions of the entries of that matrix. 
Thus, we conclude the transverse reactivity $r(\bx_s) = e_1 (M)$  is a continuous function of $\bx_s$.
The proof is complete. \hfill \qed

\subsection{Worst-case probability} \label{sec:prob}

Here we define the `worst-case' 
probability of observing an increase in the norm of the transverse perturbation at initial times by randomly selecting a point $\bx_s$ from the attractor:
\begin{equation} \label{main}
    \mu :=  \frac{1}{2} + \frac{1}{2}\Big\langle \sign \Big( r(\bx_s ) \Big)  \Big\rangle_{\mathcal{A}}.
\end{equation}
Here, $\sign(\cdot)$ is the sign function and $<\cdot>_{\mathcal{A}}$ indicates an average over the attractor $\mathcal{A}$. 
The quantity $0 \leq \mu \leq 1$ measures the fraction of the points on the attractor that result in $r(\bx_s) > 0$ for some values of $\sigma$ and $\xi$. 
The term `worst-case' refers to 
the worst possible choice of the initial condition $\delta\hat{\bX}(0)$ for Eq.\,\eqref{deltaXX}, which 
is a scalar multiple of the eigenvector corresponding to the largest eigenvalue of the matrix $(\hat{Z}(\bx_s) + \hat{Z}(\bx_s)^\top)/2$.
However, if $\delta\hat{\bX}(0)$ is chosen randomly, the initial condition will have a nonzero component along this eigenvector with probability one.
Hence, by defining $\mu$ as above, we now can provide a  
probability that an increase in $\| \delta\hat{\bX}(t) \|$ will be typically seen at the initial time.
%
\begin{proposition} \label{prop:mu}
The worst-case probability $\mu$ is a monotonically non-decreasing function of $p=\sigma \xi$.
\end{proposition}
\begin{proof}
    Since $\mu$ is a non-decreasing continuous function of $r(\bx_s)$ and $r(\bx_s)$ is a non-decreasing continuous function of $p = \sigma \xi$, we conclude $\mu$ is a non-decreasing continuous function of $p $. The proof is complete.
\end{proof}
{
Following the same steps in Sec.\,\ref{sec:syncreactivity}, it follows from Proposition \ref{prop:mu} that for two networks $A$ and $B$, if the syncreactivity $\Xi^A > \Xi^B$, then the critical probability $\mu^A \geq \mu^B$.
}

\section*{Data availability}
All data generated or analyzed during this study are included in this published article (and its supplementary information files).

\section*{Code availability}
The source code for the numerical simulations presented in the paper will be made available upon request, as the code is not required to support the main results reported in the manuscript.

 \newcommand{\noop}[1]{}

\section*{Acknowledgements}
We acknowledge support from grants AFOSR FA9550-24-1-0214 and Oak Ridge National Laboratory 006321-00001A. The authors are grateful to Chad Nathe for the work he performed on an early version of this paper and 
to Marco Storace for the thorough and generous feedback he has provided on this paper.
\section*{Author Contributions}
Amirhossein Nazerian worked on the theory and numerical simulations. Joe Hart worked on the experimental system, as well as on some aspects of the theory. Matteo Lodi worked on the simulations and the figures. Francesco Sorrentino worked on the theory and supervised the research. 
All authors contributed to writing the paper.

\section*{Competing Interests Statement}
The authors declare no competing interests

\end{document}